\documentclass{tpms-l}

\newtheorem{theorem}{Theorem}[section]
\newtheorem{lemma}[theorem]{Lemma}

\theoremstyle{definition}

\theoremstyle{remark}
\newtheorem{remark}[theorem]{Remark}

\numberwithin{equation}{section}

\usepackage{graphicx}

\begin{document}

\title{Average Precision at Cutoff k under Random Rankings: Expectation and Variance}

\author{Tetiana Manzhos}
\address{Department of Advanced Mathematics,
Kyiv National Economic University named after Vadym Hetman, Kyiv 03057, Ukraine}
\email{tmanzhos@kneu.edu.ua}

\author{Tetiana Ianevych} 
\address{Department of Probability Theory, Statistics and Actuarial Mathematics, Taras Shevchenko National University of Kyiv, Kyiv 01601, Ukraine}
\email{tetianayanevych@knu.ua}
\thanks{The second author was supported by the Ministry of Education and Science of Ukraine within the grant for the prospective development of the scientific field “Mathematical Sciences and Natural Sciences” at Taras Shevchenko National University of Kyiv.}
 
\author{Olga Melnyk}
\address{Department of Advanced Mathematics,
Kyiv National Economic University named after Vadym Hetman, Kyiv 03057, Ukraine}
\email{melnyk.olha@kneu.edu.ua}


\date{02/NOV/2025}

\dedicatory{This paper is dedicated to the 85th anniversary of the birth of Prof. Yuriy Kozachenko. 
}

\keywords{Recommendation systems, information retrieval, ranking algorithms, AP@k, MAP@k, evaluation metrics, expectation and variance, random rankings}

\begin{abstract}
Recommender systems and information retrieval platforms rely on ranking algorithms to present the most relevant items to users, thereby improving engagement and satisfaction. Assessing the quality of these rankings requires reliable evaluation metrics. Among them, Mean Average Precision at cutoff k (MAP@k) is widely used, as it accounts for both the relevance of items and their positions in the list. 

In this paper, the expectation and variance of Average Precision at k (AP@k) are derived since they can be used as biselines for MAP@k. Here, we covered two widely used evaluation models: offline and online. The expectation establishes the baseline, indicating the level of MAP@k that can be achieved by pure chance. The variance complements this baseline by quantifying the extent of random fluctuations, enabling a more reliable interpretation of observed scores. 
\end{abstract}

\maketitle

\section{Introduction} 
Recommendation systems play a pivotal role in modern digital environments, designed to predict and suggest items such as movies on streaming platforms, products on e-commerce sites,  or news articles on media platforms based on users' preferences and historical interactions.  These systems leverage sophisticated algorithms that analyze user behavior,  item attributes,  and contextual data to provide personalized recommendations, thereby enhancing user satisfaction and engagement.

Ranking algorithms are fundamental to recommendation systems, as they determine the order in which items are presented to users. These algorithms prioritize relevant items and arrange them in a sequence designed to maximize user satisfaction and engagement. Metrics that assess the effectiveness of these rankings are essential for evaluating algorithmic performance. For instance, \textit{Kendall Tau Distance} measures the similarity between the recommended ranking order and the ground truth, providing insights into the consistency and reliability of recommendation algorithms. \textit{Reciprocal Rank (RR)} and  \textit{Mean Reciprocal Rank (MRR)} focus on the position of the first relevant recommendation in the list, emphasizing the importance of early relevance in user interaction. A comprehensive review of relevant metrics is provided in \cite{arXiv:2312.16015}, offering a deeper exploration of these evaluation methods.

Among various evaluation metrics, \textit{Average Precision at $k$ (AP@k)} plays a crucial role in recommendation systems due to its ability to assess the precision of recommended items up to a specified rank $k$. It not only accounts for the relevance of the items but also their position in the recommendation list, providing a comprehensive evaluation of ranking effectiveness. AP@k is particularly important in scenarios where users are more likely to interact with top-ranked items, highlighting the necessity for recommendation algorithms to accurately prioritize relevant content. For example, in search engines like Google that are based on the PageRank method (for details, see \cite{BrinPage1998, EngstromSilvestrov2017} and the references therein), users tend to click on the top few results, making it crucial for the search algorithm to rank the most relevant results at the top.
A natural extension of AP@k is MAP@k (Mean Average Precision at $k$), which averages the AP@k scores across users. While AP@k evaluates the ranking quality for individual users, MAP@k provides an overall measure of the algorithm's performance across the entire user base, giving a more generalized assessment of the system’s ability to rank relevant items consistently.

This paper develops closed-form expressions for the expected value and the variance of AP@k under random rankings. Since MAP@k averages user-level AP@k, the expectation specifies the random baseline for the metric, whereas the variance determines the scale of random fluctuations around that baseline. Together, these results indicate the level of MAP@k that can arise fully by chance and provide a quantitative reference for interpreting how far an observed MAP@k exceeds the random baseline. This baseline supports principled benchmarking and guides the optimization of recommendation strategies.

This paper is organized as follows. Section~\ref{sec1} introduces two evaluation settings that we consider (offline and online) and explains their practical relevance. Section~\ref{sec2} provides formal definitions of AP@k and MAP@k and clarifies their relationship. Sections~\ref{sec3} and \ref{sec4} present closed-form expressions for the expectation and variance of AP@k in both offline and online scenarios and include the corresponding proofs. Section~\ref{sec5} presents some practical results and comparative tables that illustrate the particularities of AP@k in several scenarios, summarizes the main findings, and suggests directions for future research.
\section{Offline and online evaluation} \label{sec1}
In the evaluation of recommendation systems, two principal approaches are widely used: \textit{offline} evaluation and \textit{online} one. A comprehensive discussion of these methods can be found in \cite{Agarwal_Chen_2016}. Both approaches share the common principle of treating item relevance as binary –- classifying items as relevant or not –- but they differ in the methodology used to estimate the proportion of relevant items.

\textit{Offline} evaluation relies on historical data, where the relevance of items for each user is known. It is particularly useful when comprehensive, labeled data, such as historical interactions, is available, making it possible to test recommendation algorithms in a controlled environment.

In this scenario, metrics such as MAP@k (Mean Average Precision at $k$) can be computed by averaging AP@k over a subset of users. Since the relevant items for each user are predefined, the number of such items, denoted as $m$, remains fixed. In the evaluation phase, a subset of $N$ items per user is selected, with $m$ known as relevant. If we want to compare the ranking algorithm with a random one, we need to define the random experiment that can be utilized under such a scenario. In this case, one random result differs from another only in the places where $m$ relevant elements are located among $N$ available places. It corresponds to the experiment in which we sample $m$ items from $N$ with equal probabilities and without replacement.

In this context, the expected value for AP@k under the random recommendation approach serves as a natural baseline for comparison. By evaluating the system performance against this baseline, offline evaluation helps establish how much better the algorithm performs compared to random item recommendation, assuming a fixed number of relevant items $m$. This benchmarking not only highlights the algorithm’s effectiveness but also provides a foundation for assessing improvements across different methods.

Even though \textit{online} evaluation also treats item relevance as binary, it does not require a predefined number of relevant items $m$ within a set of $N$ items per user. This is because, in an online evaluation, it is impractical to determine a fixed number of relevant items for each user. Instead, it is essential to ensure that every participant in the sample is shown at least $k$ items from the recommendation list. For reliable evaluation in this scenario, two key conditions must be met: the pool of items used to generate the recommendation list must be sufficiently large, and the proportion of relevant items should remain approximately stable across all users in the treated group. These conditions enable the evaluation of the recommendation system's performance by comparing it against the expected baseline under random recommendations. This comparison is facilitated by the probability $p$, which represents the likelihood of each item being viewed (i.e., relevant). With this probability, it is possible to calculate the expected value of the evaluation metric under random item presentation, serving as an effective benchmark.

The random experiment that fits this type of evaluation is Bernoulli sampling. It means that we have $N$ items, and suppose that every item can be either relevant or not, independently of each other, but with a probability of being relevant equal to $p$. This is equivalent to sampling $m$ items from $N$, where $m=pN$ ($m$ is assumed to be natural) with equal probabilities and with replacement.

\section{Mean Average Precision at $k$}\label{sec2}
Having established the context of offline and online evaluation, we now turn to the formal definition of the Mean Average Precision at $k$ (MAP@k) metric, which is central to this paper. This metric is widely used for assessing the quality of ranked recommendation lists by taking into account both the relevance of items and their positions within the list. To begin, we first define Average Precision at $k$ (AP@k) for a single user, which serves as the building block for MAP@k.

The  Average  Precision at  $k$  measures the ranking quality of relevant items by averaging precision values at positions where relevant items appear. Formally, it is defined as:
\begin{equation}\label{eq1}
AP@k=\frac 1k \sum_{i=1}^{k}P@i \cdot rel(i),
\end{equation}
where:
\begin{itemize}
    \item $k$ is the length of the recommendation list;
    \item $P@i=\frac {\sum_{j=1}^{i}rel(j)}{i}$ is the precision at position $i$;
    \item $rel(i)\in\{0,1\}$ is a binary indicator, equal to 1 if the item at position $i$ is relevant and 0 otherwise.
\end{itemize}

In this general form, the denominator $k$ normalizes the metric over the entire recommendation list. However, when the total number of relevant items $m$ is smaller than $k$, a modified normalization factor, $\min(m,k)$, is used:
\begin{equation} \label{eq2}
AP@k=\frac {1}{\min(m,k)} \sum_{i=1}^{k}P@i \cdot rel(i).
\end{equation}

This adjustment ensures that the metric appropriately handles situations where there are fewer relevant items than the length of the recommendation list, avoiding artificially deflated scores. Without this modification, even in the best scenario – where all $m$ relevant items are ranked in the top positions – the AP@k score would never reach 1, because the denominator would be larger than $m$, preventing the score from reaching its maximum value.

Having defined Average Precision at $k$ (AP@k), which measures the ranking quality for a single user, we now extend this concept to the Mean Average Precision at $k$ (MAP@k), which aggregates the AP@k values across all users in the dataset. MAP@k provides a more comprehensive evaluation of the recommendation system by averaging the precision scores for all users, thereby reflecting the system's overall performance in ranking relevant items.
Formally, MAP@k is computed as the mean of the AP@k scores across all users:
\begin{equation} \label{eq3}
MAP@k=\frac 1{\lvert U \rvert} \sum_{u \in U}AP@k_u,
\end{equation}
where:
\begin{itemize}
    \item $U$ is the set of all users in the dataset;
    \item $AP@k_u$ is the $AP@k$ score for user $u$.
\end{itemize}

Turning to the interpretation of the MAP@k metric, a high MAP@k value indicates that most users encounter relevant items ranked highly, reflecting an effective recommendation system. Conversely, a low MAP@k suggests that relevant items appear lower, signaling a less efficient algorithm. The metric ranges from 0 to 1, where 1 means all relevant items are within the top-$k$ positions, and 0 means none are included.
However, for a comprehensive evaluation, it is essential to compare the MAP@k score with the expected value derived from random recommendation generation. While one might assume that this expected value equals the proportion of relevant items in the collection, this is not the case. As demonstrated by Yves Bestgen in his work \cite{article_Bestgen}, the expected average precision for random recommendations is more complex and varies depending on various factors such as the number of relevant items in the list and other characteristics. In particular.   In the current study, we derive precise formulas for both the offline and online evaluation scenarios described above, providing an accurate benchmark for comparison against the random baseline, under sampling without and with replacement, respectively.

\section{Expectation and variance of AP@k for the offline evaluation (sampling without replacement)}\label{sec3}

If we want to determine how much the algorithm outperforms random chance, we can compare the MAP@k of the algorithm with the expected value of AP@k if randomness is imposed by some random experiment. The expected AP@k alone is not sufficient for a comprehensive evaluation, but it provides the crucial baseline against which meaningful improvements can be interpreted. In addition, the variance quantifies the spread of AP@k values around this baseline, completing the picture by indicating the scale of random fluctuations. 

We begin with the offline evaluation setting described in Section \ref{sec1}, corresponding to sampling without replacement. In this case, among $N$ items, exactly $m$ are relevant, and one random outcome differs from another only by the positions of the $m$ relevant items in the ranking. This makes the offline setting a natural starting point for analysis. 

In this randomization model, the positions of the $m$ relevant items among the $N$ are chosen uniformly at random. 
Equivalently, the ranking can be viewed as a uniformly random permutation of all $N$ items, so that the $m$ relevant ones are placed at random positions. 
From the mathematical point of view, this means that for each recommendation $j = 1, 2, \ldots, N$ we define a binary random variable $I_j$, with $I_j = 1$ if the item at position $j$ is relevant and $I_j = 0$ otherwise. Furthermore,  $\sum_{j=1}^NI_j=m$ and the random variables $I_1, I_2, ..., I_N$ are not independent.

Under this model, the probability that the $j$-th item is relevant can be calculated as
\[
P (I_j = 1) =\frac{C_{N-1}^{m-1}}{C_N^m}=\frac{m}{N}.
\]

So,  $E(I_j) = m/N$ for every $j$ and the random variable $\sum_{j=1}^iI_j$, which counts the number of relevant items among the first $i$ items,  follows the \textit{Hypergeometric} distribution with parameters $N$, $m$ and $i$. Moreover 

$$\operatorname{E}(P@i)=\frac{1}{i}\sum_{j=1}^i\operatorname{E} (I_j)= \frac{1}{i}\sum_{j=1}^i \frac{m}{N}=\frac{m}{N}, \quad j=1,..., k.$$
In this case, $AP@k$ is computed according to formula \eqref{eq2} where the normalization is made by $\min(m,k)$:
$$AP@k=\frac{1}{\min(m,k)}\sum_{i=1}^k P@i\cdot I_{i}.$$

\begin{remark}
    For identification that the expectations are calculated with respect to different random models,  we used different subscripts (WOR and WR) in the theorems' formulation. In the proofs, we omit these subscripts to simplify the formulas. 
\end{remark}

\begin{theorem}\label{th1wor} If among  $N$ available items, exactly $m$ are relevant for every user,  then the value $MAP_{WOR}@k$ (i.e., the expectation of $AP@k$ with respect to all possible random results under sampling without replacement) is equal to
\begin{equation}
MAP_{WOR}@k= \operatorname{E}_{WOR} (AP@k) =\frac{m}{N\cdot \min (k,m)}\left(\frac{m-1}{N-1}k+\frac {N-m}{N-1} H_k\right),
\end{equation}
where $H_k=\sum_{i=1}^k\frac1i$ is the $k$-th ''harmonic'' number that can be calculated using the approximation $H_k\approx \ln k+\gamma+\frac1{2k}$, $\gamma=0.5772$ is the Euler-Mascheroni constant.
\end{theorem}

\begin{proof}
$$MAP@k=\operatorname{E}(AP@k) =\frac{1}{\min (m,k)}\sum_{i=1}^k \operatorname{E}(P@i\cdot I_i).$$
Let's first calculate the following:
\begin{equation*}
\begin{split}
\operatorname{E}(P@i\cdot I_i)&=\operatorname{E}\big(\operatorname{E}(P@i\cdot I_i|I_i)\big)=\\
&=\operatorname{E}\big(P@i\cdot I_i|I_i=1\big)P\{I_i=1\}+\operatorname{E}\big(P@i\cdot I_i|I_i=0\big)P\{I_i=0\}=\\
&=\operatorname{E}\big(P@i|I_i=1\big)\cdot\frac{m}{N}=\frac{1}{i}\bigg((i-1)\frac{m-1}{N-1}+1\bigg)\cdot\frac{m}{N},
\end{split}
\end{equation*}
since, if we know that $I_i=1$, then $P@i$ is distributed as $\frac{1}{i}\cdot \big(Hypergeometric(N-1,m-1,i-1)+1\big).$ 
Therefore,
\begin{equation*}
\begin{split}\operatorname{E}(AP@k)&=\frac{1}{\min(m,k)}\cdot\frac{m}{N}\sum_{i=1}^k \bigg((1-\frac{1}{i})\frac{m-1}{N-1}+\frac{1}{i}\bigg)=\\
&=\frac{1}{\min(m,k)}\cdot\frac{m}{N}\sum_{i=1}^k \left((1-\frac{m-1}{N-1})\frac{1}{i}+\frac{m-1}{N-1} \right)=\\
&=\frac{1}{\min(m,k)}\cdot\frac{m}{N} \left(k\cdot \frac{m-1}{N-1}+(1-\frac{m-1}{N-1})\sum_{i=1}^k \frac{1}{i} \right).
\end{split}
\end{equation*}
For calculating the sums $\sum_{i=1}^k \frac{1}{i} $, which are, actually, the partial sums of the harmonic series and were named as harmonic numbers $ H_{k}$ in 1968 by Donald Knuth \cite{Knuth1968}, we can use the well-known formula (see \cite{Boas01101971})
$$
H_{k}=\ln k+\gamma +{\frac {1}{2k}}-\varepsilon _{k},
$$
where $ \gamma \approx 0.5772$ is the Euler–Mascheroni constant and 
$ 0\leq \varepsilon _{k}\leq 1/8k^{2}$  which tends to 0 as $ 
k\to \infty$ .
This completes the proof.
\end{proof}
\begin{remark}
For small $k$, the numbers $H_k$ can be calculated exactly by ordinary summation. 
\end{remark}

While the expected value $\operatorname{E}(AP@k)$ (equivalently $MAP@k$) provides the baseline for evaluating ranking quality, the variance complements it by quantifying the extent of random fluctuations. Its explicit expression is established in the following theorem.

\begin{theorem} \label{th2wor} If we consider the offline evaluation procedure, that is, the random appearance of a fixed number $m$ of relevant items corresponds to sampling $m$ items from $N$ with equal probabilities and without replacement, then 
\begin{equation*}
\begin{split} \operatorname{Var}_{WOR}\big(AP@k\big)=\frac{1}{M^2} \frac{m}{N}\Bigg[& k(C+2(E-F) + (k-1)G) + H_k(B-2(E-kF))+\\ 
& + H_k^2 \cdot D + H_k^{\left(2\right)}(A-D) \Bigg],
\end{split}
\end{equation*}
where 
\begin{align*}
 A & =  1- \frac{m}{N}-\frac{m-1}{N-1} \left(3-2\frac{m-2}{N-2}-\frac{m}{N}\left(2-\frac{m-1}{N-1}\right)\right); \\
 B & =  \frac{m-1}{N-1}\left(3\left(1-\frac{m-2}{N-2}\right)-2\frac{m}{N}\left(1-\frac{m-1}{N-1}\right)\right);\\
C & = \frac{m-1}{N-1}\left(\frac{m-2}{N-2}-\frac{m(m-1)}{N(N-1)}\right);\\
D & = \frac{m-1}{N-1}\left(2-5\frac{m-2}{N-2}+3\frac{(m-2)(m-3)}{(N-2)(N-3)}\right)-\frac{m}{N}\left(1-\frac{m-1}{N-1}\right)^2;\\
E& = \frac{m-1}{N-1}\left(3\frac{m-2}{N-2}\left(1-\frac{m-3}{N-3}\right)-\frac{m}{N}\left(1-\frac{m-1}{N-1}\right)\right);\\
F &= \frac{m-1}{N-1}\left(\frac{m-2}{N-2}\left(1-\frac{m-3}{N-3}\right)-\frac{m}{N}\left(1-\frac{m-1}{N-1}\right)\right); 
\end{align*}
 \begin{align*}
G & = \frac{m-1}{N-1}\left(\frac{(m-2)(m-3)}{(N-2)(N-3)}-\frac{m}{N}\frac{m-1}{N-1}\right);\\
H_{k} & = \sum_{i=1}^k \frac{1}{i}; \quad H^{(2)}_{k} = \sum_{i=1}^k \frac{1}{i^2};\\
M & = \min(m,k). 
\end{align*}
\end{theorem}
In the proof, we need some technical results, which are formulated as a lemma.  We present it before the proof of Theorem~\ref {th2wor}. 

\begin{lemma}\label{lem1}
For any $k\geq 1$
\begin{itemize}
\item[(i)]  $\sum_{i=1}^{k-1}\sum_{l=i+1}^{k}\frac{1}{i}  =  k(H_k-1);$
\item[(ii)] $\sum_{i=1}^{k-1} \sum_{l=i+1}^{k} \frac{1}{l}  = k-\sum_{l=1}^{k}\frac{1}{l}=k-H_k;$
\item[(iii)] $\sum_{i=1}^{k-1} \sum_{l=i+1}^{k} \frac{1}{i\cdot l}  =  \frac{1}{2}\bigg(H_k^2-H_k^{\left(2\right)}\bigg),$   
\end{itemize}
where  $H_k=\sum_{i=1}^{k} \frac{1}{i}$ is a standard harmonic number; 
      $H_k^{\left(2\right)}=\sum_{i=1}^{k} \frac{1}{i^2}$ is the $k$-th partial sum of the 2nd order harmonic series.
\end{lemma}
\begin{proof}
(i) \[
\sum_{i=1}^{k-1}\sum_{l=i+1}^{k}\frac{1}{i} = \sum_{i=1}^{k-1}\frac{1}{i}(k-i) = k \sum_{i=1}^{k-1}\frac{1}{i} - (k-1) + (\frac{1}{k} - \frac{1}{k})k = k \sum_{i=1}^{k}\frac{1}{i} - k = k(H_k-1).
\]
(ii)  $$
     \sum_{i=1}^{k-1} \sum_{l=i+1}^{k} \frac{1}{l}=\sum_{l=2}^{k} \sum_{i=1}^{l-1} \frac{1}{l}=\sum_{l=2}^{k}\frac{l-1}{l} =k-1-\bigg(\sum_{l=1}^{k}\frac{1}{l}-1\bigg)=k-\sum_{l=1}^{k}\frac{1}{l}=k-H_k.
     $$
(iii) For the calculation of the last equality, we can use the fact  that for any finite collection of numbers $\{a_{il}\}_{i,l=1}^{k}$, such that $a_{il}=a_{li}$ if $i\neq l$ 
  $$\sum_{i=1}^{k-1} \sum_{l=1+1}^{k}a_{il}=\frac{1}{2}\bigg(\sum_{i=1}^{k} \sum_{l=1}^{k} a_{il}-\sum_{i=1}^{k}a_{ii}\bigg).$$ 
  So, 
$$\sum_{i=1}^{k-1} \sum_{l=i+1}^{k} \frac{1}{i\cdot l}= \frac{1}{2}\bigg(\sum_{i=1}^{k} \sum_{l=i+1}^{k} \frac{1}{i\cdot l}-\sum_{i=1}^{k} \frac{1}{i^2}\bigg) = \frac{1}{2}\bigg(\bigg(\sum_{i=1}^{k} \frac{1}{i}\bigg)^2 -\sum_{i=1}^{k} \frac{1}{i^2}\bigg) = \frac{1}{2}\bigg(H_k^2-H_k^{\left(2\right)}\bigg);
$$
\end{proof}

\begin{proof}[Proof of Theorem~\ref {th2wor}]
Since $\operatorname{Var}(AP@k) = \operatorname{Var}\left(\frac{1}{k} \sum_{i=1}^k P@i \cdot I_i\right) = \frac{1}{k^2} \sum_{i=1}^k \sum_{l=1}^k c_{il}$, we need to calculate  all  the coefficient $c_{il}$ taking into account that
 for every $i\neq l$
$c_{il} = c_{li} = \operatorname{cov}(P@i \cdot I_i, P@l \cdot I_l)$
and  for all $i$
$c_{ii} = \operatorname{Var}(P@i \cdot I_i)$.

Let us calculate them step by step, starting from the coefficients $c_{ii}$, $i=1,2,...,k.$ Since in this case the random variables $I_1,..., I_n$ are not independent, it is worth applying formulas that use conditional expectations. In particular, for two random variables $X$ and $Y$, we can calculate the variance by the formula:
$$\operatorname{Var}(X)=\operatorname{E}\big(\operatorname{Var}(X\mid Y)\big)+Var\big(\operatorname{E}(X\mid Y)\big).$$
So,
\begin{equation}\label{t1t2}
c_{ii}=\operatorname{Var}\big(P@i\cdot I_i\big)=\operatorname{E}\big(\operatorname{Var}(P@i\cdot I_i\mid I_i\big)+\operatorname{Var} \big(\operatorname{E}(P@i\cdot I_i\mid I_i) \big)=t_1+t_2.
\end{equation}
If $I_i=0,$ then $P@i\cdot I_i=0,$ and $\operatorname{E}\big(P@i\cdot I_i\mid I_i=0\big)=0, \operatorname{Var}\big(P@i\cdot I_i\mid I_i=0\big)=0.$ \\
If $I_i=1,$ then $\operatorname{E}\big(P@i\cdot I_i\mid I_i=1\big)=\frac{1}{i} \left((i-1)\frac{m-1}{N-1}+1 \right)$ and 
$$\operatorname{Var}\big(P@i\cdot I_i\mid I_i=1\big)=\operatorname{E}\left(\big(P@i\cdot I_i\big)^2\mid I_i=1\right)-\big(\operatorname{E}(P@i\cdot I_i\mid I_i=1)\big)^2.$$

First of all
\begin{equation*}
\begin{split}
\operatorname{E}\big((P@i\cdot I_i)^2\mid I_i=1\big)&=\operatorname{E}\left(\left(\frac{1}{i}\left(\sum_{j=1}^{i-1} I_jI_i+I_i^2\right)\right)^2\mid I_i=1\right)=\\
=&\frac{1}{i^2}\operatorname{E}\left(1+2\sum_{j=1}^{i-1} I_j
+\left(\sum_{j=1}^{i-1} I_j\right)^2\mid I_i=1\right)=\\
=&\frac{1}{i^2}\operatorname{E}\left(1+2\sum_{j=1}^{i-1} I_j
+\sum_{j=1}^{i-1} I_j^2+\sum \sum_{j\neq r}I_j I_r\mid I_i=1\right)=\\
=&\frac{1}{i^2}\left(1+3\sum_{j=1}^{i-1} \operatorname{E}\big(I_j \mid I_i=1\big)+
\sum \sum_{j\neq r}\operatorname{E}\big(I_j I_r\mid I_i=1\big)\right)=\\
=&\frac{1}{i^2}\left(1+3(i-1)\cdot \frac{m-1}{N-1}+
(i-1)(i-2)\frac{(m-1)(m-2)}{(N-1)(N-2)}\right),
\end{split}
\end{equation*}
as $\operatorname{E}\big(I_j\mid I_i=1\big)=\frac{C_{N-2}^{m-2}}{C_{N-1}^{m-1}}=\frac{m-1}{N-1}$ and $\operatorname{E}(I_j I_r \mid I_i=1)=\frac{C_{N-3}^{m-3}}{C_{N-1}^{m-1}}=\frac{(m-1)(m-2)}{(N-1)(N-2)}$.

Then \begin{equation*}
\begin{split} 
\operatorname{Var}(P@i\cdot I_i \mid I_i=1)&=\frac{1}{i^2}\left(1+3(i-1)\cdot \frac{m-1}{N-1}+
(i-1)(i-2)\frac{(m-1)(m-2)}{(N-1)(N-2)}\right)-\\
&-\frac{1}{i^2}\left((i-1)\frac{m-1}{N-1}+1 \right)^2
\end{split} 
\end{equation*}
and thus we obtain the first term in \eqref{t1t2}: 
\begin{equation*}
\begin{split} &t_1=\operatorname{E}\big(\operatorname{Var}(P@i\cdot I_i \mid I_i)\big)=\operatorname{Var}\big(P@i\cdot I_i \mid I_i=1\big)\cdot \frac{m}{N}=\\
&=\frac{m}{N} \cdot \frac{1}{i^2} \left[1+3(i-1)\cdot \frac{m-1}{N-1}+
(i-1)(i-2)\frac{(m-1)(m-2)}{(N-1)(N-2)}-((i-1)\frac{m-1}{N-1}+1)^2 \right].
\end{split}
\end{equation*}

Secondly, if $I_i=1$ the random variable $\operatorname{E}(P@i\cdot I_i \mid I_i)$ takes the value $\frac{1}{i}\cdot \big((i-1)\frac{m-1}{N-1}+1\big)$ with probability $\frac{m}{N}$ and $0$ otherwise. Therefore, 
\begin{equation*}
\begin{split} t_2&=\operatorname{Var}\big(\operatorname{E}(P@i\cdot I_i \mid I_i)\big)=\operatorname{E}\big((\operatorname{E}(P@i\cdot I_i \mid I_i))^2\big)-\big(\operatorname{E}(\operatorname{E}(P@i\cdot I_i \mid I_i))\big)^2=\\
&=\frac{1}{i^2}\left((i-1)\frac{m-1}{N-1}+1\right)^2\cdot \frac{m}{N}-\frac{1}{i^2}\left((i-1)\frac{m-1}{N-1}+1\right)^2\cdot \left(\frac{m}{N}\right)^2=\\
&=\frac{1}{i^2}\left((i-1)\frac{m-1}{N-1}+1\right)^2\cdot \frac{m}{N}\left(1-\frac{m}{N}\right).
\end{split}
\end{equation*}
So, finally, we have the following.
\begin{equation*}
\begin{split}
c_{ii}=&\operatorname{Var}\big(P@i\cdot I_i\big)=t_1+t_2=\\
= &\frac{m}{N} \cdot \frac{1}{i^2} \left(1+3(i-1)\cdot \frac{m-1}{N-1}+
(i-1)(i-2)\frac{(m-1)(m-2)}{(N-1)(N-2)}\right)-\\
&-\frac{m}{N} \cdot \frac{1}{i^2} \left((i-1)\frac{m-1}{N-1}+1\right)^2 +
\frac{1}{i^2}\left((i-1)\frac{m-1}{N-1}+1\right)^2\cdot \frac{m}{N}\left(1-\frac{m}{N}\right)=
\\
= &\frac{1}{i^2}\cdot \frac{m}{N} \left(1+3(i-1)\cdot \frac{m-1}{N-1}+
(i-1)(i-2)\frac{(m-1)(m-2)}{(N-1)(N-2)}-\right.\\
& -\left.\frac{m}{N}-2(i-1)\frac{m(m-1)}{N(N-1)} - 
 (i-1)^2\frac{m(m-1)^2}{N(N-1)^2} \right)=\\
= &\frac{1}{i^2}\cdot \frac{m}{N} \left(1- \frac{m}{N}+(i-1)\frac{m-1}{N-1}\left(3-2\frac{m}{N}+(i-2)\frac{m-2}{N-2}-(i-1)\frac{m}{N} \frac{m-1}{N-1}\right) \right)=\\
=& \frac{1}{i^2} \cdot \frac{m}{N} \left[1- \frac{m}{N}+\frac{m-1}{N-1} \left(2\frac{m-2}{N-2}-3+2\frac{m}{N}-\frac{m(m-1)}{N(N-1)}+ \right. \right.\\
&\left. \left. + i\left(3-2\frac{m}{N}-3\frac{m-2}{N-2}+2\frac{m(m-1)}{N(N-1)}\right) + i^2\left(\frac{m-2}{N-2}-\frac{m(m-1)}{N(N-1)}\right) \right) \right]=\\
= & \frac{m}{N} \left(1- \frac{m}{N}-\frac{m-1}{N-1} \left(3-2\frac{m-2}{N-2}-\frac{m}{N}\left(2-\frac{m-1}{N-1}\right)\right)\right)\frac{1}{i^2}+\\
& + \frac{m}{N}\frac{(m-1)}{(N-1)}\left(3\left(1-\frac{m-2}{N-2}\right)-2\frac{m}{N}\left(1-\frac{m-1}{N-1}\right)\right)\frac{1}{i}+\\ 
& + \frac{m}{N} \frac{(m-1)}{(N-1)}\frac{(m-2)}{(N-2)}-\left(\frac{m(m-1)}{N(N-1)}\right)^2 .
\end{split}
\end{equation*}
(ii)
For all such $i$ and $l$ that $1<i<l\leq k$
$$c_{il}=c_{li}=\operatorname{cov}(P@i\cdot I_i,P@l\cdot I_l)=\operatorname{E}(P@i\cdot I_i\cdot P@l\cdot I_l)-\operatorname{E}(P@i\cdot I_i)\cdot \operatorname{E}(P@l\cdot I_l).$$
Here, we can also use the conditional expectation for calculation
\[
\operatorname{E}(P@i I_i\cdot P@l I_l) = \operatorname{E}\big(\operatorname{E}(P@i I_i\cdot P@l I_l \mid I_i I_l)\big) 
= \operatorname{E}(P@i I_i\cdot P@l I_l \mid I_i I_l = 1)  P(I_i I_l = 1).
\]
If $I_iI_l=1$, then we know that the items in the $i$th and $l$th places are relevant, so we should take into account the possible arrangements of the other $m-2$ relevant items between the $N-2$ places. So,  
$P(I_iI_l=1)=\frac{C_{N-2}^{m-2}}{C_N^m}=\frac{m(m-1)}{N(N-1)}$. And then
\begin{equation*}
\begin{split}
\operatorname{E}(P@i\cdot & I_i\cdot P@l\cdot I_l \mid I_i I_l = 1)=\operatorname{E}\left(\frac{1}{i} \sum_{j=1}^{i} I_j 
\cdot \frac{1}{l} \sum_{r=1}^{l} I_r \mid I_i I_l = 1\right)=\\
= &\frac{1}{i\cdot l}\operatorname{E}\left(\left(1+\sum_{j=1}^{i-1}I_j\right)\left(1+\sum_{r=1}^{i-1}I_r+\sum_{r=i+1}^{l-1}I_r+1\right) \mid I_iI_l=1\right)=\\
= &\frac{1}{i\cdot l}\operatorname{E}\left(2+3\sum_{r=1}^{i-1}I_r+\sum_{r=i+1}^{l-1}I_r+\left(\sum_{j=1}^{i-1}I_j\right)^2+\sum_{j=1}^{i-1}I_j\sum_{r=i+1}^{l-1}I_r \mid I_iI_l=1\right)=\\
=&\frac{1}{i\cdot l}\operatorname{E}\left(2+4\sum_{r=1}^{i-1}I_r+\sum_{r=i+1}^{l-1}I_r+2\sum_{j=1}^{i-r}\sum_{r=j+1}^{i-1}I_jI_r+\sum_{j=1}^{i-1}\sum_{r=i+1}^{l-1}I_jI_r \mid I_iI_l=1\right)=\\
=&\frac{1}{i\cdot l}\left(2+4(i-1)\frac{m-2}{N-2}+(l-1-i)\frac{m-2}{N-2}+(i-1)(i-2)\frac{(m-2)(m-3)}{(N-2)(N-3)}+\right.\\
& \left. +(i-1)(l-1-i)\frac{(m-2)(m-3)}{(N-2)(N-3)}\right)=\\
= &\frac{1}{i\cdot l}\left(2+\frac{m-2}{N-2}(3i-5+l)+\frac{(m-2)(m-3)}{(N-2)(N-3)}(i-1)(l-3)\right).
\end{split}
\end{equation*}
Therefore, 
\begin{equation*}
\begin{split}
c_{i l} & = \frac{m(m-1)}{N(N-1)}\frac{1}{i\cdot l}\left[2-5\frac{m-2}{N-2}+i\cdot 3\frac{m-2}{N-2}+l\frac{m-2}{N-2}+i\cdot l\frac{(m-2)(m-3)}{(N-2)(N-3)}- \right.\\
&\left.-i\cdot 3\frac{(m-2)(m-3)}{(N-2)(N-3)}-l\frac{(m-2)(m-3)}{(N-2)(N-3)}+3\frac{(m-2)(m-3)}{(N-2)(N-3)}\right]-\\
&-\left(\frac{m}{N}\right)^2\frac{1}{i\cdot l}\left(1+(i-1)\frac{m-1}{N-1}\right)\left(1+(l-1)\frac{m-1}{N-1}\right)=\\
= & \frac{m}{N}\frac{1}{i\cdot l}\left[\frac{m-1}{N-1}\left(2-5\frac{m-2}{N-2}+3\frac{(m-2)(m-3)}{(N-2)(N-3)}+i\cdot 3\frac{m-2}{N-2}\left(1-\frac{m-3}{N-3}\right)+ \right.\right.\\
&\left. +l\frac{m-2}{N-2}\left(1-\frac{m-3}{N-3}\right)+i\cdot l\frac{(m-2)(m-3)}{(N-2)(N-3)}\right)-\frac{m}{N}\left(1+l\frac{m-1}{N-1}-\frac{m-1}{N-1}+ \right.\\
&+i\frac{m-1}{N-1}-\frac{m-1}{N-1} \left.\left.+il\left(\frac{m-1}{N-1}\right)^2-l\left(\frac{m-1}{N-1}\right)^2-i\left(\frac{m-1}{N-1}\right)^2+\left(\frac{m-1}{N-1}\right)^2\right)\right]=\\
= &\frac{m}{N}\frac{1}{i\cdot l}\left[\frac{m-1}{N-1}\left(2-5\frac{m-2}{N-2}+3\frac{(m-2)(m-3)}{(N-2)(N-3)}\right)-\frac{m}{N}\left(1-\frac{m-1}{N-1}\right)^2+ \right.\\
&+i\frac{m-1}{N-1}\left(3\frac{m-2}{N-2}\left(1-\frac{m-3}{N-3}\right)-\frac{m}{N}\left(1-\frac{m-1}{N-1}\right)\right)+\\
&+ l\frac{m-1}{N-1}\left(\frac{m-2}{N-2}\left(1-\frac{m-3}{N-3}\right)-\frac{m}{N}\left(1-\frac{m-1}{N-1}\right)\right)+\\
&+ \left. i\cdot l\frac{m-1}{N-1}\left(\frac{(m-2)(m-3)}{(N-2)(N-3)}-\frac{m}{N}\frac{m-1}{N-1}\right)\right].
\end{split}
\end{equation*}
Using the notations from the theorem's statement and  Lemma~\ref{lem1},  we'll get:
\begin{equation*}
\begin{split}
\operatorname{Var}&_{WOR}(AP@k)= \frac{1}{M^2} \Bigg( \sum_{i=1}^{k}c_{ii} + 2\sum_{i=1}^{k-1}\sum_{l=i+1}^{k}c_{il} \Bigg) =
\\
= &  \frac{1}{M^2} \Bigg[\sum_{i=1}^{k}\frac{m}{N}\Bigg(\frac{A}{i^2}+\frac{B}{i}+C \Bigg)+2\sum_{i=1}^{k-1}\sum_{l=i+1}^{k}\frac{m}{N}\Bigg(\frac{D}{i\cdot l}+\frac{E}{l}+\frac{F}{i}+G \Bigg) \Bigg] =
\\
= & \frac{1}{M^2} \frac{m}{N}\Bigg[ A \cdot H_k^{\left(2\right)} + B \cdot H_k + kC + D(H_k^2 - H_k^{\left(2\right)})  +\\
& +   2Ek(H_k-1) + 2F(k-H_k) + Gk(k-1) \Bigg] =
\\
= & \frac{1}{M^2} \frac{m}{N}\Bigg[ k[C+2(E-F) + (k-1)G] + H_k(B-2(E-kF))+\\
&+ H_k^2 \cdot D + H_k^{\left(2\right)}(A-D) \Bigg]. 
\end{split}
\end{equation*}
This completes the proof.
\end{proof}

\section{Expectation and variance of AP@k for the online evaluation (sampling with replacement)} \label{sec4}

We now turn to the online evaluation setting described in Section~\ref{sec1}. Here, we assume that each item is relevant to a user with a constant probability $p$. The randomness can be modeled as a sequence of independent Bernoulli trials, or equivalently as sampling with replacement, where each item has probability $p$ of being sampled. In this formulation, the decision whether a recommendation is relevant is made independently for each item. Moreover, we assume a uniform relevance probability $p$ for all recommendations and all users. These probabilistic assumptions differ from the offline case and lead to distinct formulas for the expectation and variance of AP@k.
 
From a mathematical point of view, this means that for each recommendation $j=1, 2, ..., N$, we define the random variables $I_j=1$ if it was relevant and $I_j=0$ otherwise, as in the previous case. But here the random variables $I_1, I_2,..., I_N$ are independent with mean $\operatorname{E}(I_j) = p$ and variance $\operatorname{Var} (I_j) = p(1-p)$, where $p$ is assumed to be known and represents the expected proportion of relevant items for each user. Of course, this assumption is quite strong and can be applied only when the users are relatively similar (homogeneous), as described in Section~\ref{sec1} for the online evaluation scenario.

Under these assumptions and notations, we have the following:

$$P@i=\frac{1}{i}\sum_{j=1}^iI_j, \quad \operatorname{E}(P@i)=\frac{1}{i}\sum_{j=1}^i\operatorname{E} (I_j)= \frac{1}{i}\sum_{j=1}^i p=p,$$
and $$AP@k=\frac{1}{k}\sum_{i=1}^k P@i\cdot I_{i}.$$

For sampling with replacement, the expected $AP@k$ for random recommendations is presented in the following theorem. 

\begin{theorem}\label{th1wr} If, for each user, every item is assumed to be independently relevant with probability $p$, then the value $ MAP_{WR}@k $ (the expectation of $ AP@k $ under sampling with replacement) is equal to
\begin{equation}\label{eq5}
MAP_{WR}@k= \operatorname{E}(AP@k) =p\left(p+\left(1-p\right)\frac 1kH_k\right), 
\end{equation}
where $H_k=\sum_{i=1}^k\frac1i$ is the $k$-th harmonic number as well as in Theorem~\ref{th1wor}.
\end{theorem}

\begin{proof}

From the definition of $AP@k$ we have
\begin{equation*}
\begin{split}
MAP@k & =   \operatorname{E} (AP@k) = \operatorname{E}\bigg(\frac 1k \sum_{i=1}^k P@i\cdot I_{i}\bigg)= \\
&=\frac 1k \sum_{i=1}^k \operatorname{E}\bigg(\frac{1}{i}\sum_{j=1}^iI_j\cdot I_{i}\bigg)= \frac 1k \sum_{i=1}^k \bigg(\frac{1}{i}\sum_{j=1}^i\operatorname{E}( I_j\cdot I_{i})\bigg). 
\end{split}
\end{equation*}
For all $j\neq i$, $\operatorname{E}( I_j\cdot I_{i})= \operatorname{E}( I_j)\cdot \operatorname{E}(I_{i})=p^2$, as these random variables are  independent. If $j=i$, then  $\operatorname{E}(I_{i}^2)= \operatorname{E}(I_{i})=P\{I_i=1\}=p$. 
So, 
\begin{equation*}
\begin{split}
MAP@k&=\frac 1k \sum_{i=1}^k \frac{1}{i}\bigg(\sum_{j=1}^{i-1}\operatorname{E}( I_j\cdot I_{i})+\operatorname{E}(I_i)\bigg)=\frac 1k \sum_{i=1}^k \frac{1}{i}((i-1)p^2+p)=\\
&=\frac pk  \sum_{i=1}^k (p+(1-p)\frac 1i)= p^2+p(1-p)\frac 1k \sum_{i=1}^k \frac{1}{i}.
\end{split}
\end{equation*}
This completes the proof. 
\end{proof}

\begin{theorem}\label{th2wr}
    Under the assumptions of Theorem~\ref{th1wr}, the variance of the $AP@k$ can be calculated using the following formula
$$
\operatorname{Var}_{WR}(AP@k) = \frac{5}{k} \cdot p^3 (1-p) + \frac{1}{k^2} \cdot p(1-p) \left[p(1-2 p)(3 H_k + H_k^2) + (1-p)(1-3 p)H_k^{\left(2\right)} \right].
$$
Here $H_k=\sum_{i=1}^k\frac1i$ is the $k$-th harmonic number and $H_k^{\left(2\right)}=\sum_{i=1}^k\frac{1}{i^2}$ is the $k$-th partial sum of the 2nd order harmonic series.
\end{theorem}

\begin{proof}
As before,   $\operatorname{Var}(AP@k) = \operatorname{Var}\left(\frac{1}{k} \sum_{i=1}^k P@i \cdot I_i\right) = \frac{1}{k^2} \sum_{i=1}^k \sum_{l=1}^k c_{il}$, and we need to calculate  the coefficient $c_{il}$ taking into account that  for every $i\neq l$
$c_{il} = c_{li} = \operatorname{cov}(P@i \cdot I_i, P@l \cdot I_l)$
and  for all $i$ 
$c_{ii} = \operatorname{Var}(P@i \cdot I_i).$

Let us calculate them again one by one. \\
(i)
For $i=1,2,...,k$
\begin{equation*}
\begin{split}
c_{ii}& = \operatorname{Var} \left( \frac{1}{i} \sum_{j=1}^{i} I_j I_i \right) = \frac{1}{i^2} \operatorname{Var} \left( \sum_{j=1}^{i-1} I_j I_i + I_i^2 \right)=\\
&=\frac{1}{i^2} \operatorname{Var} \left( \sum_{j=1}^{i-1} I_j I_i + I_i \right) = \frac{1}{i^2} \operatorname{Var} \left( I_i \left( 1 + \sum_{j=1}^{i-1} I_j \right) \right).
\end{split}
\end{equation*}
Here,  the random variables $I_{i}$ and $1 + \sum_{j=1}^{i-1} I_j$ are independent. The variance of the product of two independent random variables $X$ and $Y$ can be calculated using the following rule:
$$
\operatorname{Var}(X \cdot Y) = \operatorname{E}(X^2)\operatorname{E}\big(Y^2\big) - \big(\operatorname{E}(X))^2(\operatorname{E}(Y)\big)^2.
$$
As far as $\operatorname{E}\big(I_i^2\big) = \operatorname{E}(I_i) = p$, 

$\operatorname{E}\left(1 + \sum_{j=1}^{i-1} I_j\right) = 1 + \sum_{j=1}^{i-1} \operatorname{E}\big(I_j\big) = 1 + (i-1)p,$
and
\begin{equation*}
\begin{split}
\operatorname{E}&\left( \left( 1 + \sum_{j=1}^{i-1} I_j \right)^2 \right) =
\operatorname{E}\left( 1 + 2\sum_{j=1}^{i-1} I_j + \left( \sum_{j=1}^{i-1} I_j \right)^2 \right) =\\
&= \operatorname{E}\left( 1 + 2\sum_{j=1}^{i-1} I_j + \sum_{j=1}^{i-1} \sum_{r=1}^{i-1} I_j I_r \right) =\\
&= \operatorname{E}\left( 1 + 2\sum_{j=1}^{i-1} I_j + \sum_{j=1}^{i-1} I_j^2 + 2 \sum_{j=1}^{i-2} \sum_{r=j+1}^{i-1} I_j I_r \right) =\\
&= 1 + 3(i-1)p + (i-1)(i-2)p^2,
\end{split}
\end{equation*}
then
\begin{equation*}
\begin{split}
c_{ii}& = \frac{1}{i^2} \operatorname{Var} \left( I_i \cdot \left(1 + \sum_{j=1}^{i-1} I_j \right)\right) =\\
&= \frac{1}{i^2} \left( \operatorname{E} \left( I_i^2 \right) \operatorname{E}\left( \left(1 + \sum_{j=1}^{i-1} I_j \right)^2 \right) - \left( \operatorname{E}(I_i) \right)^2 \left(\operatorname{E} \left( 1 + \sum_{j=1}^{i-1} I_j \right) \right)^2 \right)=\\
& = \frac{1}{i^2} \left( p \left(1 + 3(i-1)p + (i-1)(i-2)p^2 \right) - p^2 \left(1 + (i-1)p \right)^2 \right)=\\
& = \frac{1}{i^2} \left( p + (3i-4)p^2 + (i-1)(i-4)p^3 - (i-1)^2p^4 \right)=\\
& = p \left( \frac{1}{i^2} + \frac{3}{i}p - \frac{4}{i^2}p + p^2 - \frac{5}{i}p^2 + \frac{4}{i^2}p^2 - p^3 + \frac{2}{i}p^3 - \frac{1}{i^2}p^3 \right)=
\\
& = p\left( p^2(1-p) + \frac{1}{i}p(3-2p)(1-p) + \frac{1}{i^2}(p^2 - 3p + 1)(1-p) \right)=
\\
& = p(1-p) \left( p^2 + p(3-2p)\frac{1}{i} + (1-3p+p^2)\frac{1}{i^2} \right).
\end{split}
\end{equation*}
(ii)  For all $i$ and $l$ for which $1<i<l\leq k$
$$c_{il} = c_{li} = \operatorname{cov}(I_i \cdot P@i, I_l \cdot P@l) = \operatorname{E}(I_i \cdot P@i \cdot I_l \cdot P@l) - \operatorname{E}(I_i \cdot P@i) \cdot \operatorname{E}(I_l \cdot P@l).$$
Here
\begin{equation*}
\begin{split}
\operatorname{E}&(I_i \cdot P@i \cdot I_l \cdot P@l) = \frac{1}{i\cdot l} \operatorname{E}\Big( I_i \Big(1+\sum_{j=1}^{i-1} I_j\Big) \cdot I_l \Big(1+\sum_{r=1}^{l-1} I_r\Big) \Big) =
\\
&= \frac{1}{i\cdot l} \operatorname{E}\Big( I_i I_l \Big(1+\sum_{j=1}^{i-1} I_j\Big) \Big(1+\sum_{r=1}^{i-1} I_r + I_i + \sum_{r=i+1}^{l-1} I_r\Big) \Big) =
\\
& = \frac{1}{i\cdot l} \operatorname{E}\Big( I_i I_l \Big(1+\sum_{j=1}^{i-1} I_j\Big)^2 + I_i^2 I_l \Big(1+\sum_{j=1}^{i-1} I_j\Big) + I_i I_l \Big(1 + \sum_{j=1}^{i-1} I_j\Big) \Big(\sum_{r=i+1}^{l-1} I_r\Big) \Big) =
\\
&= \frac{p^2}{i\cdot l} \Big( 1+3(i-1)p+(i-1)(i-2)p^2 + 1+(i-1)p +  (1+(i-1)p)(l-1-i)p \Big) =
\\
& = \frac{p^2}{i\cdot l} \Big( 2+(3i-5+l)p + (i-1)(l-3)p^2  \Big).
\end{split}
\end{equation*}
Therefore, 
\begin{equation*}
\begin{split}
c_{il} & = \frac{p^2}{i\cdot l} \Big( 2+(3i-5+l)p + (i-1)(l-3)p^2   -  \Big(1+(i-1)p\Big)\Big(1+(l-1)p\Big)\Big) =
\\
&= \frac{p^2}{i\cdot l}  \Big( 2 - 1 + p(3i-l-5-i+1-l+1) + p^2(i-1)(l-3-l+1) \Big) =
\\
&= \frac{p^2}{i\cdot l}  \Big( 1 + i \cdot 2p - 3p - i \cdot 2p^2 + 2p^2 \Big)
=  \frac{1}{l} p^2(1-p) \Big( (1-2p)\frac{1}{i} + 2p \Big).
\end{split}
\end{equation*}
Summing up all the coefficients $\{c_{il}\}_{i,l=1}^k$ we obtain
\begin{equation*}
\begin{split}
\sum_{i=1}^{k}  \sum_{l=1}^{k}  c_{il} = &  \sum_{i=1}^{k} c_{ii} + 2\sum_{i=1}^{k-1} \sum_{l=i+1}^{k} c_{il} =
\\
= & \sum_{i=1}^{k} \Bigg( p(1-p) \Big( p^2 + p (3-2p) \frac{1}{i} + (1-3p+p^2) \frac{1}{i^2} \Big) \Bigg) +
\\
&+ 2\sum_{i=1}^{k-1} \sum_{l=i+1}^{k} p^2(1-p) \Bigg( (1-2p) \frac{1}{i\cdot l} + 2p \frac{1}{l} \Bigg) =
\\
= & k p^3(1-p) + p^2(3-2p) \sum_{i=1}^{k} \frac{1}{i} + p(1-3p+p^2) \sum_{i=1}^{k} \frac{1}{i^2} +
\\
&+ 2p^2(1-p)(1-2p) \sum_{i=1}^{k-1} \sum_{l=i+1}^{k} \frac{1}{i\cdot l} + 4p^3(1-p) \sum_{i=1}^{k-1} \sum_{l=i+1}^{k} \frac{1}{l}.
\end{split}
\end{equation*}

Using the results from  Lemma~\ref{lem1} we, finally, get
\begin{equation*}
\begin{split}
\operatorname{Var}& (AP@k) = \frac{1}{k^2} \Bigg( \sum_{i=1}^{k}c_{ii} + 2\sum_{i=1}^{k-1}\sum_{l=i+1}^{k}c_{il} \Bigg) =
\\
= &\frac{1}{k^2} \Bigg( p(1-p) \sum_{i=1}^{k} \Big( p^2 + p(3-2p) \frac{1}{i} + (1-3p+p^2) \frac{1}{i^2} \Big) +
\\
&+ 2p^2(1-p) \sum_{i=1}^{k-1} \sum_{l=i+1}^{k} \Big( (1-2p) \frac{1}{i\cdot l} + 2p \frac{1}{l} \Big) \Bigg) =
\\
=& \frac{1}{k^2} p(1-p) \Bigg( k p^2 + p(3-2p) \sum_{i=1}^{k} \frac{1}{i} + (1-3p+p^2) \sum_{i=1}^{k} \frac{1}{i^2} + 
\\
&+ 2p(1-2p) \sum_{i=1}^{k-1} \sum_{l=i+1}^{k} \frac{1}{i \cdot l} + 4p^2 \sum_{i=1}^{k-1} \sum_{l=i+1}^{k} \frac{1}{l} \Bigg) =
\\
=&\frac{1}{k} \big( p^3(1-p) + 4p^3(1-p) \big) 
+ \frac{1}{k^2} p(1-p) \Big( H_k \big( p(3-2p) - 4p^2 \big) +
\\
& + H_k^{(2)} \big( 1 - 3p + p^2 - p(1 - 2p) \big) + H_k^2 p(1 - 2p) \Big) =
\\
=& \frac{5}{k}p^3(1-p)+\frac{1}{k^2}p(1-p)\left(p(1-2p)\big(3H_k+H_k^2\big)+(1-p)(1-3p)H_k^{(2)}\right).
\end{split}
\end{equation*} 
\end{proof}

\section{Practical outcomes and discussion} \label{sec5}
Before presenting some practical results, it is helpful to summarize the key differences in the formulas for the offline and online settings. In the \textit{offline} case, the distribution of AP@k depends on the total number of items $N$, the number of relevant items $m$, and the cut-off $k$ that specifies how many top-ranked positions are taken into account. This reflects the fact that evaluation is performed on a finite set of items, where exactly $m$ out of $N$ are relevant and their positions in the ranking determine the outcome. In the \textit{online} setting, by contrast, the parameter $N$ does not appear in the formulas, since evaluation is restricted to the top-$k$ recommendations and each of them is modeled as an independent Bernoulli trial with relevance probability $p$. This makes the outcome depend only on $k$ and $p$, regardless of the overall size of the candidate pool, which in practice may be very large. Conceptually, the probability $p$ plays a role analogous to the ratio $m/N$ in the offline case: while offline evaluation fixes the proportion of relevant items deterministically, online evaluation assumes it probabilistically.

To cover a representative range of conditions, $N=50$  is fixed and several choices of $m$ and $k$ are examined. The half-density case ($m=25$) with varying cut-offs illustrates the situations $m>k$, $m=k$, and $m<k$. Lower prevalences ($m=10$ and $m=2$) are also included along with a high-prevalence setting ($m=35$, corresponding to $p>0.5$). These scenarios span the spectrum from very sparse to very dense relevance distributions, thereby testing the formulas under diverse conditions. The selected configurations are summarized in Table~\ref{table1}.
\begin{center}
   	\begin{table}[ht] \footnotesize
		\caption{\textbf{Examined scenarios}}\vspace{-8pt}
		\centering
		\begin{tabular}{|c|c|c|c|c|} 
			\hline
			\textbf{Scenario}   	& \textbf{Offline: $(N,m)$}& \textbf{Online: $p$} &\textbf{ $k$} &\textbf{ Characteristic }\\
			
			\hline
			A1   & (50, 25) 			& 	0.50			 & 5		   & $m>k$, balanced case (50\% relevant)	\\ \hline
			A2	& (50, 25) 			&	 0.50			 & 25 	   & $m=k$, balanced case (50\% relevant)  	\\ \hline
			A3	& (50, 25) 			&	0.50			 & 40		   & $m<k$, balanced case (50\% relevant)	\\
			\hline
			B   & (50, 10) 			& 	0.20			 & 20		   & moderate proportion of relevant items	\\ \hline
			C	& (50, 2) 			&	 0.04			 & 20 	   & very low proportion of relevant items  	\\ \hline
			D	& (50, 35) 			&	0.70			 & 20		   & high proportion of relevant items $(p > 0.5)$	\\
            \hline
		\end{tabular}
		\label{table1}
	\end{table}
\end{center}
\vspace{10pt}

Table~\ref{table2} reports the theoretical (derived from the formulas in Sections ~\ref{sec3} and ~\ref{sec4})  means and variances of AP@k for the selected scenarios.  The offline and online settings are placed side by side, allowing examination of their behavior under comparable conditions and observation of where the two randomization schemes produce similar outcomes and where differences emerge.
\begin{center}
\begin{table}[ht] \footnotesize
    \caption{\centering \textbf{Values of expectation and variance for AP@k for scenarios from Table~\ref{table1}}}
    \vspace{-8pt}
    \centering
    \begin{tabular}{|c|c|c|c|c|} 
        \hline
        \textbf{Scenario}    & \textbf{WOR Expectation}& \textbf{WR Expectation} &\textbf{WOR Variance } &\textbf{WR Variance}\\
        \hline
        A1 & 0.36139  & 0.36416  & 0.05464  & 0.05884  \\
        \hline
        A2 & 0.28387  & 0.28816  & 0.00735  & 0.01234  \\
        \hline
        A3 & 0.43550  & 0.27674  & 0.00699  & 0.00775  \\
        \hline
        B  & 0.13221  & 0.06878  & 0.00786  & 0.00294 \\
        \hline
        C  & 0.07865  & 0.00851  & 0.01563  & 0.00023  \\
        \hline
        D  & 0.52426  & 0.52778  & 0.01502  & 0.02195  \\
        \hline
    \end{tabular}
    \label{table2}
\end{table}
\end{center}

As seen in Table~\ref{table2},  the comparison between offline (WOR) and online (WR) shows no uniform pattern: in some scenarios the two models behave very similarly, while in others their outcomes diverge noticeably.  The observed differences reflect the $\min(m,k)$ normalization in offline evaluation and the independence structure of online sampling. Together, these results provide a clear basis for interpreting MAP@k under the random rankings.

In addition to the values reported in Table~\ref{table2}, Figures~\ref{fig-1} and \ref{fig-2} show the simulated distributions of AP@k for two representative scenarios. These examples illustrate how WOR and WR can differ not only in their means but also in their variability.

\begin{figure}[tb]
    \includegraphics[width=0.6\textwidth]{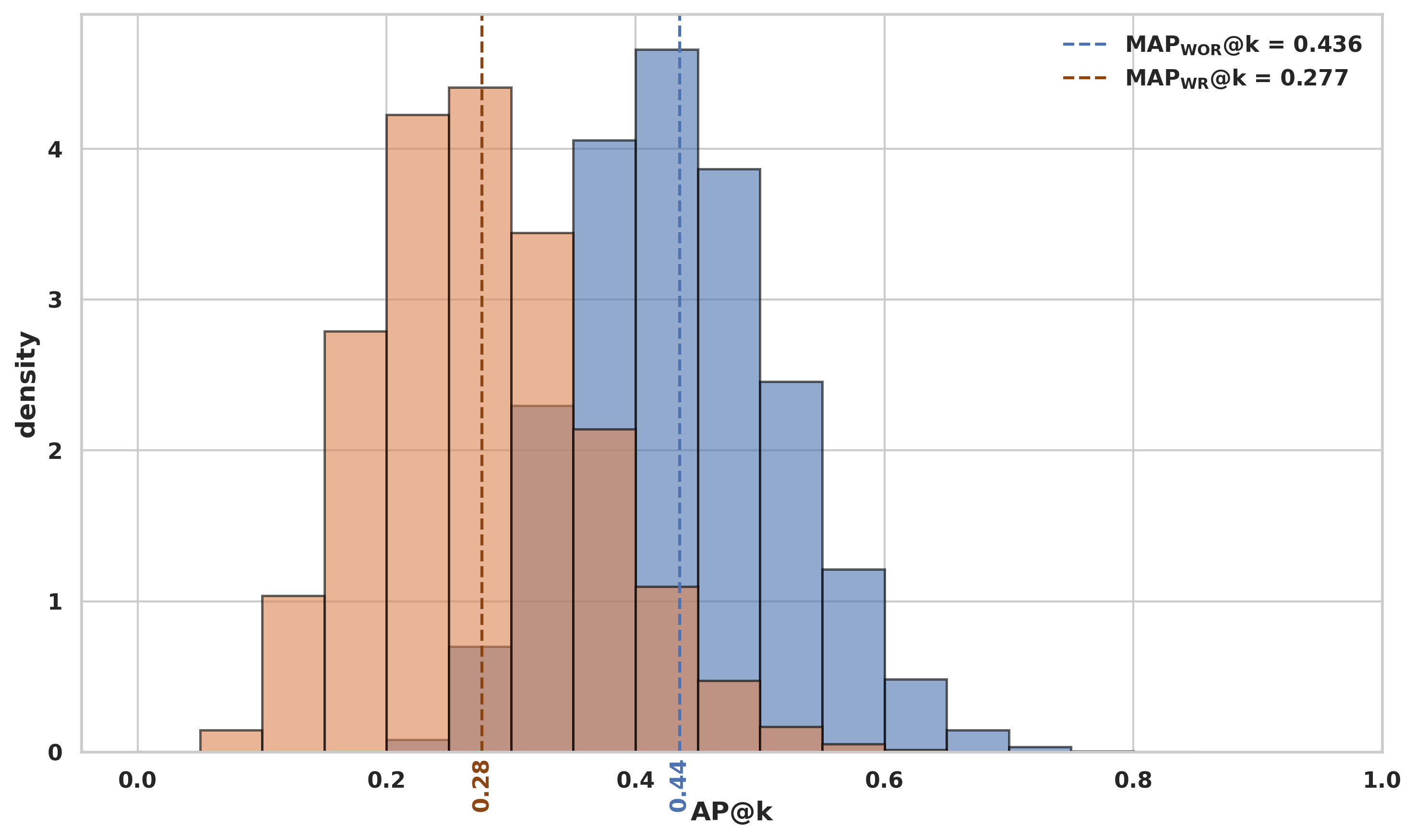} 
    \caption{\footnotesize  Histogram of AP@k values for Scenario A3 $(N = 50, m = 25, p = 0.5, k = 40)$. The WOR distribution (blue) is shifted to higher values compared to WR (orange), reflecting normalization by $\min(m, k)$ in the offline case }
    \label{fig-1}
\end{figure}

\begin{figure}[tb]
    \includegraphics[width=0.6\textwidth]{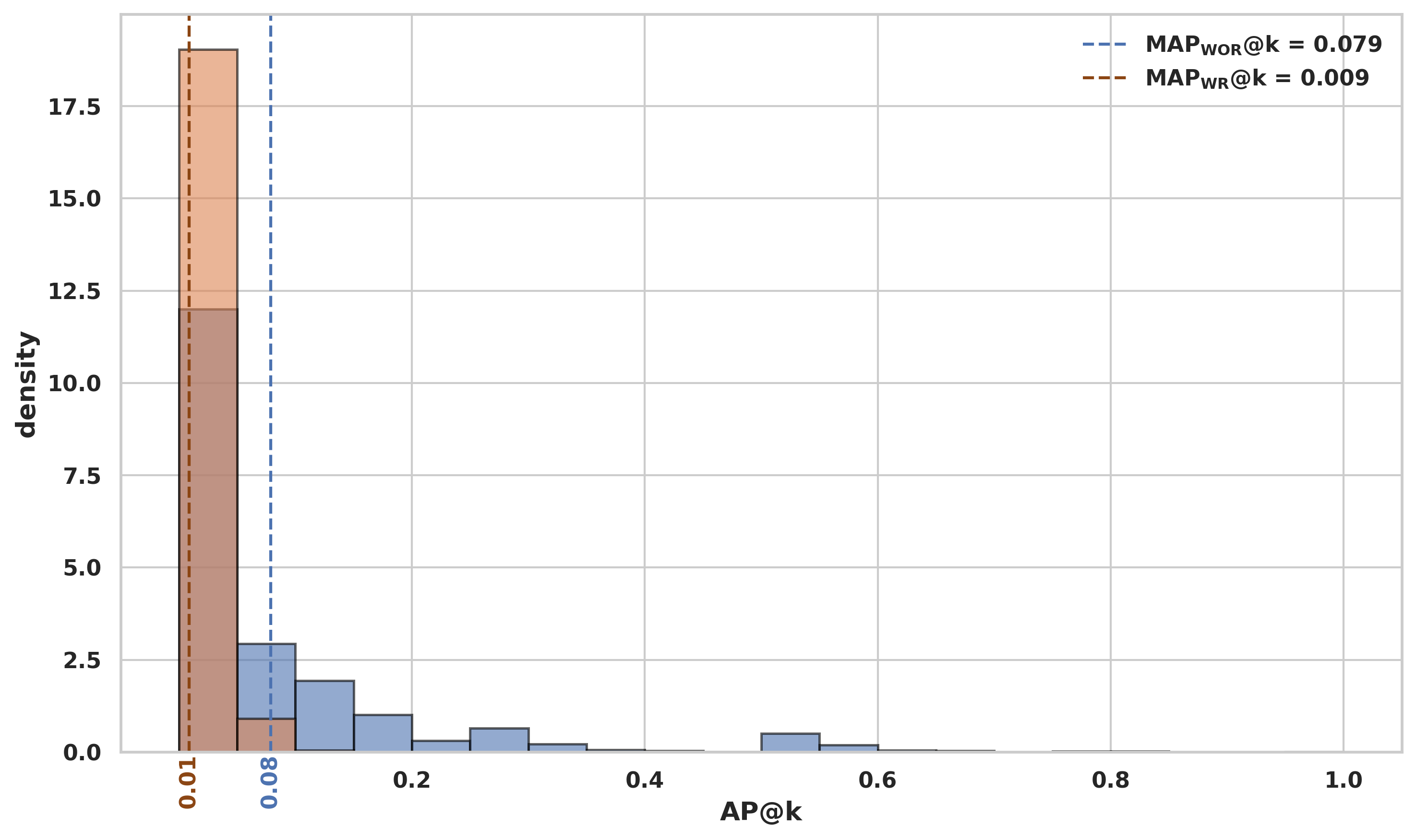}
     \caption{ \footnotesize  Histogram of AP@k values for Scenario C $(N = 50, m = 2,  p = 0.04, k = 20)$. The WR distribution (orange) is tightly concentrated around zero, while the WOR distribution (blue) shows wider variability due to the finite number of relevant items}
    \label{fig-2}
\end{figure}

In \textit{Scenario A3} ($m < k$), the distributions are shifted relative to each other: the offline case yields systematically higher values due to normalization by $\min(m,k)$, while the online case, which models relevance probabilistically, results in lower means.

In \textit{Scenario C} (very low prevalence), the contrast lies in variability. The WR distribution is sharply concentrated, reflecting that almost always no relevant items appear in the top-$k$, whereas the WOR distribution remains wider because the fixed number of relevant items can still occupy different positions in the ranking.

 These results clarify what levels of MAP@k can be expected by chance and how much variability surrounds them. From a practical standpoint, these findings offer two key insights. First, the expected value of AP@k serves as a clear reference point against which observed performance can be compared. Second, the variance quantifies the random fluctuations around this reference, helping to distinguish genuine improvements from chance outcomes. Together, expectation and variance enhance the interpretability of MAP@k and support a more principled benchmarking of recommendation algorithms.

Future research may extend this analysis to related metrics and investigate how statistical testing frameworks can build on these results for a robust comparison of ranking systems.

\bibliographystyle{amsplain}

\end{document}